\documentclass[10pt]{article}
\usepackage[margin=1in]{geometry}
\usepackage{amsmath,amssymb,amsthm,mathtools}
\usepackage{accents}
\usepackage{microtype}
\usepackage{enumitem}
\usepackage{aliascnt}
\usepackage{booktabs}
\usepackage[hidelinks]{hyperref}
\usepackage[nameinlink,noabbrev]{cleveref}
\newtheorem{theorem}{Theorem}[section]
\newaliascnt{lemma}{theorem}
\newtheorem{lemma}[lemma]{Lemma}
\aliascntresetthe{lemma}
\newaliascnt{proposition}{theorem}

\aliascntresetthe{proposition}
\theoremstyle{definition}
\newaliascnt{definition}{theorem}
\newtheorem{definition}[definition]{Definition}
\aliascntresetthe{definition}
\crefname{theorem}{Theorem}{Theorems}
\crefname{lemma}{Lemma}{Lemmas}
\crefname{proposition}{Proposition}{Propositions}
\crefname{definition}{Definition}{Definitions}
\newcommand{\OO}{\mathrm{O}}
\newcommand{\dist}{\operatorname{dist}}
\newcommand{\wdiam}{\operatorname{wdiam}}
\newcommand{\BF}{\textnormal{\textsc{BF-Dijkstra}}}
\newcommand{\Decompose}{\textnormal{\textsc{Decompose}}}
\newcommand{\Solve}{\textnormal{\textsc{Solve}}}
\newcommand{\Peel}{\textnormal{\textsc{Peel}}}
\newcommand{\Residual}{\textnormal{\textsc{Residual}}}
\newcommand{\peel}{\mathrm{peel}}
\newcommand{\res}{\mathrm{res}}
\newcommand{\Gaux}{\widehat G}
\newcommand{\Paux}{\widehat P}
\newcommand{\calS}{\mathcal S}
\newcommand{\lvec}[1]{\accentset{\scriptscriptstyle\leftarrow}{#1}}
\newcommand{\rvec}[1]{\accentset{\scriptscriptstyle\rightarrow}{#1}}
\setlist[enumerate]{leftmargin=*,itemsep=3pt,topsep=5pt}
\title{Size-Sensitive Padded Decompositions for Faster\\
Deterministic Negative-Weight Shortest Paths}
\author{Khoi Duong\\University of Minnesota\\\texttt{duong245@umn.edu}}
\date{}

\begin{document}
\maketitle
\begin{abstract}
We give a deterministic algorithm for single-source shortest paths in directed graphs with integral edge weights at least $-W$ that runs in time
\[
 \OO((m+n\log\log n)\log^2 n\log(nW)).
\]
This improves the deterministic bound of Li by a factor of $\log n$.
\end{abstract}

\section{Introduction}

The negative-weight single-source shortest-path problem on integer-weighted directed graphs admits near-linear-time algorithms, but the fastest bounds still depend on randomization. Bernstein, Nanongkai, and Wulff-Nilsen~\cite{BNWN25} combined scaling with directed low-diameter decompositions (LDDs) to obtain the first near-linear-time algorithm. Bringmann, Cassis, and Fischer~\cite{BCF23} reduced its polylogarithmic overhead. Li, Mowry, and Rao~\cite{LMR26} obtained $\OO((m+n\log\log n)\log n\log\log n\log(nW))$ time using a decomposition whose loss depends on the size of the recursive instance.

These algorithms rely on LDDs, where the only known fast algorithms are randomized. A deterministic almost-linear-time algorithm follows from the minimum-cost flow algorithm of van den Brand et al.~\cite{vdBCK+23}. Deterministic near-linear-time algorithms were later obtained by Haeupler, Jiang, and Saranurak~\cite{HJS26} and by Li~\cite{Li26}. The faster of the two, Li's algorithm, uses overlapping padded decompositions and runs in $\OO((m+n\log\log n)\log^3 n\log(nW))$ time. We improve its running time by a factor of $\log n$, leaving a gap of $\log n/\log\log n$ to the randomized bound of~\cite{LMR26}.

We work in the standard word-RAM model.

\begin{theorem}\label{thm:main}
Let $G$ be a directed graph with $n$ vertices, $m$ edges, and integral edge weights at least $-W$, where $W\ge1$. There is a deterministic algorithm that either returns a negative-weight cycle or computes single-source shortest paths in time
\[
 \OO((m+n\log\log n)\log^2 n\log(nW)).
\]
\end{theorem}

\subsection{Techniques}

We adapt the size-dependent tradeoff of Li, Mowry, and Rao~\cite{LMR26}, which builds on Seymour's approach~\cite{Sey95}. Smaller recursive pieces make more progress toward terminating the recursion, so we allow them greater overhead: a piece with core mass $b$ inside a parent of mass $M$ may pay roughly $\log(M/b)$ per unit mass. A padded decomposition has two kinds of overhead. Padding duplicates vertices near the boundaries of the pieces, while its width controls the bound on potentially negative transitions in the auxiliary graph used to combine child solutions. We apply the size-dependent allowance to both. With $L=\Theta(\log n)$ and $d$ the distance scale of the parent, a piece with $\log(M/b)\approx t$ may have duplicated boundary mass about $(t/L)b$ and padding width about $d/(Lt)$.

Two difficulties arise. First, a comparison walk may traverse children with different padding widths, so a uniform bound for a single merge would be governed by the narrowest padding it encounters. We therefore peel balls in stages with decreasing mass thresholds and geometrically narrowing padding, and merge the stages in reverse order, narrowest padding first. Second, the vertices left after peeling may still carry most of the parent's mass. We cover this residual by pieces of mass at most $M/2^\tau$, where $\tau=\Theta(\log L)$; their small size pays for a uniform padding width of $\Theta(d/(L\log L))$.

When residual division fails, it exposes a large set of small weak diameter. Such a set either gives diameter reduction directly or identifies a new peeling center. The sets that supply new centers are pairwise disjoint, so there are fewer than $2^\tau\le\sqrt L$ restarts. This keeps the quadratic cost of rescanning the centers within the $\OO(ML)$ construction budget.

\subsection{Related work}

\paragraph{Scaling.} Scaling algorithms for shortest paths with integer weights go back to Gabow~\cite{Gab85}. Gabow and Tarjan~\cite{GT89} obtained $\OO(m\sqrt n\log(nW))$ time for negative weights, and Goldberg~\cite{Gol95} improved this to $\OO(m\sqrt n\log W)$ using the scaling step that the near-linear-time algorithms, including ours, build on.

\paragraph{Strongly polynomial algorithms.} Algorithms whose running time is independent of $W$ apply to arbitrary real weights. The Bellman--Ford algorithm takes $\OO(mn)$ time. Fineman~\cite{Fin24} gave the first improvement, running in $\tilde\OO(mn^{8/9})$ time. Huang, Jin, and Quanrud improved this to $\tilde\OO(mn^{4/5})$~\cite{HJQ25} and then to $\tilde\OO(mn^{3/4}+m^{4/5}n)$~\cite{HJQ26}. Further improvements were obtained by Li, Li, Rao, and Zhang~\cite{LLRZ26}, whose algorithm runs in $\OO(n^{2.5}\log^{4.5}n)$ time, and by Quanrud and Tajkhorshid~\cite{QT26}, whose algorithm runs in $\tilde\OO(mn^{1/\sqrt2})$ time using a subroutine of~\cite{LLRZ26}.

\subsection{Organization}
\Cref{sec:prelim} describes the scaling framework and the accounting potential. \Cref{sec:merge} combines valid potentials on an ordered cover, at a cost governed by how often walks ascend in the order. \Cref{sec:decomp} constructs the size-sensitive padded decomposition and its merge order, and \cref{sec:sssp} assembles the recursion.

\section{Framework and preliminaries}\label{sec:prelim}

\subsection{Overview}

Scaling reduces the problem to $\OO(\log(nW))$ steps, each of which halves the magnitude of the most negative reduced edge weight (\cref{sec:scaling}). Following~\cite{BNWN25,BCF23,LMR26}, a step adds half the current bound to every edge weight and seeks a potential that makes these shifted weights nonnegative. To find this potential, we recursively cover the current vertex set by overlapping subsets, following Li~\cite{Li26}. The decomposition uses clipped lengths, which replace negative shifted weights by zero. Each child either has mass smaller by a constant factor or has weak diameter at most half the current distance scale.

We solve the induced child graphs recursively and combine their potentials using \BF{} (\cref{sec:merge}). Padding bounds the number of potentially negative transitions on short comparison walks in the auxiliary graph. The resulting potential at the root completes the scaling step. If a merge instead produces a negative cycle or a walk certificate, we recover a negative cycle in the input graph (\cref{sec:sssp}).

\subsection{Graphs and potentials}

Let $G=(V,E,w)$ be the input graph, with $n=|V|$ and $m=|E|$. For a directed weighted graph $F$, write $F[S]$ for the subgraph induced by $S$, $n_F$ and $m_F$ for its numbers of vertices and edges, and $\dist_F$ for directed distance. When the weights of $F$ are nonnegative, set
\[
 B_F^+(s,r)=\{v:\dist_F(s,v)\le r\},\qquad
 B_F^-(s,r)=\{v:\dist_F(v,s)\le r\}.
\]
The weak diameter of $S\subseteq V(F)$ is $\wdiam_F(S)=\max_{u,v\in S}\dist_F(u,v)$; paths witnessing this bound may leave $S$. The length of a walk counts repeated edges with multiplicity.

A vertex potential $\phi$ changes an edge weight to $w_\phi(u,v)=w(u,v)+\phi(u)-\phi(v)$. It is valid for $w$ if all these reduced weights are nonnegative. Given a valid potential, a nonnegative SSSP computation followed by undoing the potential recovers the original distances~\cite{Joh77}.

\subsection{Scaling and BF-Dijkstra}\label{sec:scaling}

We use the scaling reduction of Goldberg~\cite{Gol95} in the formulation used by Bernstein, Nanongkai, and Wulff-Nilsen~\cite{BNWN25}.

\begin{lemma}[Scaling]\label{lem:scaling}
For integral weights at least $-W$, it suffices to implement a scaling step that returns either a negative-weight cycle or a potential with reduced weights at least $-W/2$. There are $\OO(\log(nW))$ steps, followed by a nonnegative SSSP computation, to obtain exact shortest paths.
\end{lemma}

Fix one scaling step. From now on, $w$ denotes the current reduced weights, so $G=(V,E,w)$ has $w(e)\ge-W$. Define the shifted weights and the clipped graph by
\begin{equation}\label{eq:scale}
 c(e)=w(e)+W/2,\qquad \ell(e)=\max\{c(e),0\},\qquad
 \bar G=(V,E,\ell).
\end{equation}
Thus $c(e)\ge-W/2$, and a valid $c$-potential gives the required lower bound $-W/2$ for reduced $w$-weights, since
\[
 w(u,v)+\phi(u)-\phi(v)=c(u,v)+\phi(u)-\phi(v)-W/2.
\]
We construct the decomposition using distances in $\bar G$, and solve for valid $c$-potentials on the induced subgraphs of $G$ corresponding to its children. When decomposing a vertex set $X$, we write $H=\bar G[X]$. For $S\subseteq X$, the graphs $G[S]$ and $H[S]$ have the same vertices and edges; we measure walks in them by $c$ and by $\ell$, respectively.

After solving the children, the merge uses the Bellman--Ford/Dijkstra hybrid~\cite{DI17,BCF23}. We use the comparison formulation stated by Li~\cite{Li26}.

\begin{lemma}[BF-Dijkstra]\label{lem:bfd}
Let $F$ be a directed weighted graph with $n_F=n^{\OO(1)}$, let $s$ be a source, and let $\eta\ge1$ be an integer. There is a deterministic algorithm that runs in time
\[
 \OO(\eta\,(m_F+n_F\log\log n))
\]
and returns either exact distances from $s$ or, for some vertex $x$, an $s$-to-$x$ walk strictly shorter than every $s$-to-$x$ walk with at most $\eta$ negative-weight edges.
\end{lemma}

\subsection{Mass and accounting}

We define a mass measure for graph searches and vertex duplication. Let $\deg_G(v)$ count both incoming and outgoing edges and let $\lambda=\lceil\log\log n\rceil$. Set
\begin{equation}\label{eq:mass}
 \omega(v)=\deg_G(v)+\lambda,\qquad M(S)=\sum_{v\in S}\omega(v).
\end{equation}
A Dijkstra search exploring $S$ costs $\OO(M(S))$ with the integer priority queue~\cite{Tho03}. A cover of a vertex set $S$ is a family $(Y_i)$ of subsets of $S$ whose union is $S$; its members may overlap. Each copy of a vertex retains its mass $\omega(v)$, so the cover's total mass is $\sum_i M(Y_i)$, which we call its \emph{incidence mass}.

Write
\begin{equation}\label{eq:global}
 M_0=M(V)=\Theta(m+n\log\log n),\qquad
 L=\lceil\log M_0\rceil=\Theta(\log n).
\end{equation}
A recursive instance $(X,d)$ carries a vertex set $X$ and a distance scale $d$. A child makes progress by reducing its mass or by halving $d$ with a corresponding weak-diameter bound; measuring progress by weak diameter follows~\cite{FHL+25}. We will bound the work of constructing a decomposition and combining the child solutions by $\OO(L)$ times the decrease in an accounting potential. These decreases telescope over the recursion. Define
\begin{equation}\label{eq:psi}
 \Psi(z,d)=z(\log z+\log(4d/W)),\qquad \Psi(0,d)=0.
\end{equation}
Children with smaller mass or a halved parameter have smaller potential.

\section{Ordered merge}\label{sec:merge}

Given valid potentials on overlapping children, we seek a potential on their union. The cost of the merge depends on how often a walk must move to a later child in a chosen order.

Let $\mathcal Y=(Y_1,\ldots,Y_k)$ be an ordered cover of $S$.\footnote{We consider covers with $\sum_i|Y_i|=n^{\OO(1)}$, as holds for every cover used below.} A \emph{$\mathcal Y$-coloring} of a walk $P=(v_0,\ldots,v_r)$ assigns a color $\chi(t)\in\{1,\ldots,k\}$ to each position $0\le t\le r$, with $v_t\in Y_{\chi(t)}$. Different occurrences of the same vertex may receive different colors. Write $\mathcal C_{\mathcal Y}(P)$ for the set of all such colorings. A position $t<r$ is an \emph{ascent} of $\chi$ if $\chi(t)<\chi(t+1)$, and $\chi$ is \emph{ascent-free} if it has no ascents. The \emph{ascent number} of $P$ is
\[
 a_{\mathcal Y}(P)=
 \min_{\chi\in\mathcal C_{\mathcal Y}(P)}
 |\{0\le t<r:\chi(t)<\chi(t+1)\}|.
\]

\begin{theorem}[Ordered merge]\label{thm:merge}
Let $\mathcal Y=(Y_1,\ldots,Y_k)$ cover $S$, with total incidence mass $I$, and suppose a valid $c$-potential is given on each $G[Y_i]$. Let $d>0$ and let $q$ be a positive integer, and suppose that every walk $P$ in $G[S]$ with $\ell(P)\le d$ satisfies $a_{\mathcal Y}(P)\le q$. There is a deterministic algorithm that, in $\OO(qI)$ time, computes a valid $c$-potential on $G[S]$, finds a negative $c$-cycle in $G[S]$, or returns a walk $P$ in $G[S]$ satisfying
\begin{equation}\label{eq:certificate}
 c(P)<0,\qquad \ell(P)>d.
\end{equation}
\end{theorem}

\subsection{Shifting the child potentials}

Let $\phi_i$ be the given potential on $Y_i$. We first add a constant $\alpha_i\ge0$ to each $\phi_i$ so that every edge from a child to an earlier child, or to the same child, has nonnegative reduced weight.

Process the children in increasing order and maintain
\[
 \phi_{<i}(v)=\max_{j<i:\,v\in Y_j}\phi_j(v)
\]
for every $v$ that lies in some earlier child, where each $\phi_j$ has already been shifted. Choose
\begin{equation}\label{eq:offset}
 \alpha_i=\max\left\{0,\max_{\substack{u\in Y_i,\ (u,v)\in E(G[S])\\
                                  v\in Y_1\cup\cdots\cup Y_{i-1}}}
                 (\phi_{<i}(v)-c(u,v)-\phi_i(u))\right\},
\end{equation}
set $\phi_i\leftarrow\phi_i+\alpha_i$, and update the maxima on $Y_i$. After all shifts, every edge $(u,v)$ with $u\in Y_i$, $v\in Y_j$, and $i\ge j$ satisfies
\begin{equation}\label{eq:backward}
 c(u,v)+\phi_i(u)-\phi_j(v)\ge0.
\end{equation}
For $i=j$, adding the same constant to every potential value in $Y_i$ leaves internal reduced edge weights unchanged. For $i>j$, the inequality follows from \eqref{eq:offset}, since $\phi_{<i}(v)\ge\phi_j(v)$. Scanning incident edges once per membership of the source vertex computes all offsets in $\OO(I)$ time.

\subsection{Transition graph}

For each vertex $v\in S$ and each child $Y_i$ containing it, create a copy $v_i$. Conceptually, each original edge $e=(u,v)$ should join every source copy $u_i$ to every destination copy $v_j$, with weight
\begin{equation}\label{eq:conceptual}
 c(e)+\phi_i(u)-\phi_j(v).
\end{equation}
Adding all these edges would be quadratic in the number of copies. Instead, we build an auxiliary graph $\Gaux$ that uses two shared chains for each destination vertex $v$. From a source copy $u_i$, one chain reaches the destination copies $v_j$ with $j\le i$; the other reaches those with $j>i$. Only the second route may require a negative edge.

Let $\mu(v)$ be the number of children containing $v$, and let $j_1<\cdots<j_{\mu(v)}$ be their indices. Define
\[
 \lvec{\phi}_t(v)=\max_{h\le t}\phi_{j_h}(v),\qquad
 \rvec{\phi}_t(v)=\max_{h\ge t}\phi_{j_h}(v).
\]
For $1\le t\le\mu(v)$, create chain vertices $\lvec{v}_t$ and $\rvec{v}_t$. The prefix chain runs toward smaller indices and has the following edges, with weights in parentheses:
\[
 \lvec{v}_t\longrightarrow \lvec{v}_{t-1}
   \quad\left(\lvec{\phi}_t(v)-\lvec{\phi}_{t-1}(v)\right),\qquad
 \lvec{v}_t\longrightarrow v_{j_t}
   \quad\left(\lvec{\phi}_t(v)-\phi_{j_t}(v)\right),
\]
where the first edge exists for $t>1$. The suffix chain runs toward larger indices and has edges
\[
 \rvec{v}_t\longrightarrow \rvec{v}_{t+1}
   \quad\left(\rvec{\phi}_t(v)-\rvec{\phi}_{t+1}(v)\right),\qquad
 \rvec{v}_t\longrightarrow v_{j_t}
   \quad\left(\rvec{\phi}_t(v)-\phi_{j_t}(v)\right),
\]
where the first edge exists for $t<\mu(v)$. The edges to $v_{j_t}$ let a path leave either chain at that copy. All chain edges are nonnegative.

For each original edge $e=(u,v)$ and source copy $u_i$, let $t^-$ be the last position with $j_{t^-}\le i$ and $t^+$ the first with $j_{t^+}>i$. Whenever the respective positions exist, add the \emph{backward connector} and the \emph{forward connector}
\begin{align}
 u_i&\longrightarrow \lvec{v}_{t^-}
       \quad\left(c(e)+\phi_i(u)-\lvec{\phi}_{t^-}(v)\right),\label{eq:prefix-connector}\\
 u_i&\longrightarrow \rvec{v}_{t^+}
       \quad\left(c(e)+\phi_i(u)-\rvec{\phi}_{t^+}(v)\right).\label{eq:suffix-connector}
\end{align}
Through the backward connector and the prefix chain, $u_i$ reaches every destination copy $v_j$ with $j\le i$; through the forward connector and the suffix chain, it reaches every copy $v_j$ with $j>i$. In both cases the chain weights telescope to $\phi_j(v)$, so the total weight is \eqref{eq:conceptual}. The maximum $\lvec{\phi}_{t^-}(v)$ is attained at a child with index at most $i$, so the backward connector is nonnegative by \eqref{eq:backward}, and every backward route uses only nonnegative edges. On a forward route, only the connector can have negative weight.

The copies $v_{j_t}$ and the chain vertices $\lvec{v}_t$ and $\rvec{v}_t$ all represent the original vertex $v$. To express the edge weights uniformly, define
\[
 \zeta(v_{j_t})=\phi_{j_t}(v),
 \quad \zeta(\lvec{v}_t)=\lvec{\phi}_t(v),
 \quad \zeta(\rvec{v}_t)=\rvec{\phi}_t(v).
\]
Add a source $s_\ast$ with an edge of weight $-\zeta(x)$ to every copy and chain vertex $x$. Consider a walk $\Paux=(s_\ast,x_0,\ldots,x_p)$ in $\Gaux$ with $x_p=x$. It \emph{represents} the walk $P$ in $G[S]$ obtained by replacing each connector with its original edge and omitting the source and chain edges. The represented walk ends at the original vertex represented by $x$; if $\Paux$ has no connectors, $P$ consists of this single vertex. Each chain edge $x_h\to x_{h+1}$ has weight $\zeta(x_h)-\zeta(x_{h+1})$, and each connector has that same difference plus $c(e)$ for its original edge $e$. The connector contributions $c(e)$ sum to $c(P)$, so
\begin{equation}\label{eq:represented-weight}
 \begin{aligned}
 w_{\Gaux}(\Paux)
 &= -\zeta(x_0)+c(P)
    +\sum_{h=0}^{p-1}(\zeta(x_h)-\zeta(x_{h+1}))\\
 &= -\zeta(x_0)+c(P)+\zeta(x_0)-\zeta(x_p)\\
 &= c(P)-\zeta(x).
 \end{aligned}
\end{equation}
For a closed walk avoiding $s_\ast$, all $\zeta$-differences cancel and $w_{\Gaux}(\Paux)=c(P)$. A negative cycle in $\Gaux$ therefore represents a negative closed walk in $G[S]$, from which we can extract a negative $c$-cycle.

The graph has $\OO(\sum_v\mu(v))$ vertices and edges other than connectors, and at most two connectors per original edge and source copy. Build each vertex's list of child indices by visiting the children in order. For an edge $(u,v)$, process the copies $u_i$ in increasing order of $i$ and advance a pointer through $v$'s list to locate $t^-$ and $t^+$. The two lists require $\OO(\mu(u)+\mu(v))$ work for this edge, so the total construction time is
\[
 \OO\!\left(\sum_{v\in S}\mu(v)+
 \sum_{(u,v)\in E(G[S])}(\mu(u)+\mu(v))\right)=\OO(I).
\]
Since $\omega(v)\ge\lambda$, the graph also satisfies
\begin{equation}\label{eq:gadget-size}
 m_{\Gaux}+n_{\Gaux}\log\log n=\OO(I).
\end{equation}

\subsection{Walks with few negative edges}

\BF{} compares walks ending at the same auxiliary vertex. We therefore need a walk with few negative edges to the same endpoint, including when that endpoint is a chain vertex.

\begin{lemma}[Preserving the endpoint]\label{lem:preserving-endpoint}
Assume the hypothesis of \cref{thm:merge}. Let $\Paux:s_\ast\leadsto x$ be a walk in $\Gaux$ representing a walk $P$ in $G[S]$ with $\ell(P)\le d$. Then there is an $s_\ast$-to-$x$ walk of weight at most $w_{\Gaux}(\Paux)$ with at most $q+3$ negative edges.
\end{lemma}
\begin{proof}
If $P$ has no edges, the direct source edge to $x$ has the same weight as $\Paux$ and at most one negative edge. Otherwise, let $e=(u,v)$ be the last edge of $P$, and let $u_i$ be the tail of the connector in $\Paux$ representing this occurrence of $e$. Delete the final edge occurrence from $P$ to obtain a walk $P'$ ending at $u$. Since $\ell(P')\le d$, choose a $\mathcal Y$-coloring of $P'$ with at most $q$ ascents. Recolor the last position of $P'$ with color $i$; this creates at most one additional ascent, at the last edge. Construct an auxiliary walk representing $P'$: start with the source edge to the copy given by the first color, and between consecutive positions follow the connector and chain route from the current copy to the copy given by the next color. Each ascent contributes at most one negative edge, its forward connector, and the source edge contributes at most one more, so this walk ends at $u_i$ and uses at most $q+2$ negative edges.

Append the last connector and remaining chain route of $\Paux$. This adds at most one negative edge and ends at $x$. The resulting auxiliary walk represents $P$ and ends at $x$, so it has exactly the weight of $\Paux$ by \eqref{eq:represented-weight}.
\end{proof}

\begin{proof}[Proof of \cref{thm:merge}]
Construct $\Gaux$ and run \BF{} with $\eta=q+3$. If it returns exact distances, define
\[
 \Phi(v)=\dist_{\Gaux}(s_\ast,v_i)+\phi_i(v)
\]
using any child $Y_i$ containing $v$. This value is the minimum $c$-weight of a walk in $G[S]$ ending at $v$, including the trivial walk. Indeed, every auxiliary walk represents such a walk, and every walk in $G[S]$ is represented by an auxiliary walk ending at any prescribed copy of its last vertex; by \eqref{eq:represented-weight}, both have the same corrected weight. Thus the value is independent of the copy and satisfies $\Phi(v)\le\Phi(u)+c(u,v)$ on every edge.

If \BF{} instead returns a walk $\Paux:s_\ast\leadsto x$, it is shorter than the direct source edge, which uses at most one negative edge. Thus $w_{\Gaux}(\Paux)<-\zeta(x)$, so the walk $P$ represented by $\Paux$ has $c(P)<0$. If $\ell(P)\le d$, \cref{lem:preserving-endpoint} gives a comparison walk of weight at most $w_{\Gaux}(\Paux)$, contradicting the guarantee for $\Paux$. Otherwise $P$ satisfies \eqref{eq:certificate}. Constructing $\Gaux$, running \BF{}, and processing the returned walk cost $\OO(qI)$ by \eqref{eq:gadget-size}.
\end{proof}

\section{Padded decomposition}\label{sec:decomp}

Given a vertex set $X\subseteq V$ and a parameter $d$, we seek a cover of $X$ whose members either have mass at most a constant fraction of $M=M(X)$ or have weak diameter at most $d/2$. Throughout this section, $H=\bar G[X]$.

\begin{theorem}[Padded decomposition]\label{thm:decomp}
For every $X\subseteq V$ and $W/2\le d\le M_0^2W/2$, there is a deterministic algorithm that, in $\OO(ML)$ time, constructs a cover $(Y_i)$ of $X$ and assigns each member a parameter $d_i\in\{d,d/2\}$. The cover has incidence mass $\OO(M)$. For a constant $0<\rho<1$, each member satisfies one of the following conditions:
\begin{itemize}[leftmargin=*,itemsep=3pt,topsep=5pt]
\item \textup{Size reduction:} $d_i=d$ and $M(Y_i)\le\rho M$.
\item \textup{Diameter reduction:} $d_i=d/2$ and $\wdiam_H(Y_i)\le d/2$.
\end{itemize}
The decomposition decreases the accounting potential by
\begin{equation}\label{eq:actual-drop}
 \Pi_X=\Psi(M,d)-\sum_i\Psi(M(Y_i),d_i)=\Omega(M).
\end{equation}
\end{theorem}

The children are solved recursively, each $G[Y_i]$ with parameter $d_i$, and their solutions are then combined as follows.

\begin{theorem}[Reconstruction]\label{thm:reconstruction}
For the decomposition constructed in \cref{thm:decomp}, given valid $c$-potentials on the graphs $G[Y_i]$, there is a deterministic algorithm that, in $\OO(L\Pi_X)$ time, computes a valid $c$-potential on $G[X]$, finds a negative $c$-cycle in $G[X]$, or returns a walk $P$ in $G[X]$ satisfying \eqref{eq:certificate}.
\end{theorem}

\subsection{Construction overview}\label{sec:decomp-overview}

The construction alternates peeling and residual division. Peeling removes ball cores around stored centers and keeps their pads as children. Residual division applies the local decomposition repeatedly to the remaining vertices, seeking a cover by pieces of mass at most a fixed cutoff. If it succeeds, these pieces and the peeled pads form the final cover.

Otherwise, residual division returns a set above the cutoff with small weak diameter, together with a center. We test that center in the parent graph: it either gives a decomposition with diameter reduction or supplies a light direction for peeling. In the latter case, we store the center and direction and restart from the full parent set. The exposed sets are pairwise disjoint, which bounds the number of attempts.

For a successful residual division, reconstruction first merges its cover, then the peeled stages in reverse order. \Cref{sec:ascents} bounds walk ascents under padded splits, and \cref{sec:overlap-accounting} bounds the potential consumed by duplicated vertices. \Cref{sec:local} gives the local decomposition used by residual division in \cref{sec:residual-cover}. \Cref{sec:peeling} describes peeling and its staged merge, \cref{sec:construction} combines the subroutines, and \cref{sec:reconstruction} proves the reconstruction bound.

\subsection{Walk ascents}\label{sec:ascents}

\begin{definition}[Padded split]\label{def:padded-split}
Let $Y_1,Y_2\subseteq S\subseteq X$ with $Y_1\cup Y_2=S$. The ordered pair $(Y_1,Y_2)$ is a \emph{$\delta$-padded split} of $H[S]$, for $\delta>0$, if
\[
 \dist_{H[S]}(u,v)\ge\delta
 \qquad(u\in Y_1\setminus Y_2,\ v\in Y_2\setminus Y_1).
\]
\end{definition}

\begin{definition}[Padded ball]\label{def:padded-ball}
Let $S\subseteq X$. A \emph{$\delta$-padded ball} in $H[S]$, with center $s\in S$, radius $r\ge0$, direction $\sigma\in\{+,-\}$, and width $\delta>0$, is the pair
\[
 B^\circ=B_{H[S]}^\sigma(s,r),\qquad B=B_{H[S]}^\sigma(s,r+\delta).
\]
We call $B^\circ$ its \emph{core}, $B$ its \emph{pad}, and $\partial B=B\setminus B^\circ$ its \emph{shell}.
\end{definition}

An outward $\delta$-padded ball $(B^\circ,B)$ in $H[S]$ gives the padded split $(B,S\setminus B^\circ)$. Here $Y_1\setminus Y_2=B^\circ$ and $Y_2\setminus Y_1=S\setminus B$. Every vertex of $B^\circ$ has distance at most $r$ from $s$, and every vertex of $S\setminus B$ has distance greater than $r+\delta$, so the triangle inequality gives distance greater than $\delta$ from the former to the latter. An inward padded ball gives the split $(S\setminus B^\circ,B)$ by reversing directions.

\begin{definition}[Refinement]\label{def:refinement}
A \emph{refinement} of an ordered cover replaces one member $A$ by an ordered cover of $A$, inserting the new members consecutively in the position of $A$.
\end{definition}

\begin{lemma}[Padding composition]\label{lem:composition}
Let $S\subseteq X$ and $\delta>0$. Suppose $\mathcal Y$ is obtained from $(S)$ by finitely many refinements, each replacing a member $A$ by a $\delta$-padded split of $H[A]$. Then every walk $P$ in $H[S]$ satisfies
\[
 a_{\mathcal Y}(P)\le\lfloor\ell(P)/\delta\rfloor.
\]
\end{lemma}
\begin{proof}
First let $P$ be a walk in $H[S]$ of length less than $\delta$. Initially every position has color 1. Maintain an ascent-free coloring with respect to the current list. When a member $A$ is replaced by a $\delta$-padded split $(Y_1,Y_2)$ of $H[A]$, each maximal block of positions carrying its color is a walk in $H[A]$ of length less than $\delta$. Recolor each such block without ascents as follows. If the block stays inside $Y_2$, give every position the color of $Y_2$. Otherwise, use the color of $Y_2$ before the first vertex outside $Y_2$, and the color of $Y_1$ from that vertex onward. That vertex lies in $Y_1\setminus Y_2$, and the rest of the block stays inside $Y_1$, since a walk from $Y_1\setminus Y_2$ to $Y_2\setminus Y_1$ has length at least $\delta$. Since $Y_1$ and $Y_2$ occupy consecutive positions in the list, in that order, no ascent is created within or at the ends of a block. Induction over the refinements gives $a_{\mathcal Y}(P)=0$.

For a general walk $P$, greedily form blocks of positions whose internal edges have total length less than $\delta$. Cut at an edge when including it would bring the length to $\delta$ or more, and start the next block at its head. Each completed block together with its following cut edge has length at least $\delta$, and these edge sets are disjoint, so there are at most $\lfloor\ell(P)/\delta\rfloor$ cuts. Color each block without ascents. Each cut creates at most one ascent, so $a_{\mathcal Y}(P)\le\lfloor\ell(P)/\delta\rfloor$.
\end{proof}

A $\delta'$-padded split with $\delta'\ge\delta$ is also $\delta$-padded, so \cref{lem:composition} applies to refinements of different widths, with $\delta$ the smallest width.

\subsection{Accounting for overlaps}\label{sec:overlap-accounting}

The children of a decomposition overlap, and overlaps increase their total potential. If the children instead partitioned the parent and retained parameter $d$, their masses $b_i$ would sum to $M$, and the potential would decrease by
\[
 \Psi(M,d)-\sum_i\Psi(b_i,d)
 =\sum_i b_i\log(M/b_i).
\]
Thus splitting into smaller children permits more work, and halving a child's parameter further decreases its potential by its mass. To measure how much of this decrease the overlaps consume, we assign each vertex to exactly one child containing it.

\begin{definition}[Core partition]\label{def:core-partition}
A \emph{core partition} of a cover $(Y_i)$ of $S$ is a partition $S=\bigsqcup_i Y_i^\circ$ with $Y_i^\circ\subseteq Y_i$; empty cores are allowed. We call $Y_i^\circ$ the \emph{core} of $Y_i$ and $\partial Y_i=Y_i\setminus Y_i^\circ$ its \emph{shell}, and write $b_i=M(Y_i^\circ)$. The cover's \emph{duplicated mass} is
\[
 D_S=\sum_i M(\partial Y_i)=\sum_i M(Y_i)-M(S).
\]
\end{definition}

For a parent instance $(S,d)$ and child parameters $d_i\in\{d,d/2\}$, define the \emph{potential decrease before padding} and the \emph{actual potential decrease} by
\begin{equation}\label{eq:before-padding}
 \Pi^\circ_S:=\Psi(M(S),d)-\sum_i\Psi(b_i,d_i),\qquad
 \Pi_S:=\Psi(M(S),d)-\sum_i\Psi(M(Y_i),d_i).
\end{equation}
Expanding \eqref{eq:psi} and using $\sum_i b_i=M(S)$ gives
\begin{equation}\label{eq:core-credit}
 \Pi^\circ_S
 =\sum_{i:b_i>0}b_i\log(M(S)/b_i)+\sum_{i:d_i=d/2}b_i.
\end{equation}
We call the first sum the \emph{size credit} $\kappa_S$ and the second the \emph{diameter credit} of the core partition. The size credit rewards children with small cores, and the diameter credit rewards halving the parameter. The next lemma bounds how much padding reduces this decrease.

\begin{lemma}[Padding cost]\label{lem:padding-cost}
Let $(Y_i^\circ)$ be a core partition of a cover $(Y_i)$ of a nonempty set $S\subseteq X$, with duplicated mass $D_S$. Let $W/2\le d\le M_0^2W/2$ and $d_i\in\{d,d/2\}$. Then
\begin{equation}\label{eq:core-padding}
 \Pi^\circ_S-\Pi_S\le6LD_S.
\end{equation}
\end{lemma}
\begin{proof}
By \eqref{eq:before-padding},
\[
 \Pi^\circ_S-\Pi_S=\sum_i(\Psi(M(Y_i),d_i)-\Psi(b_i,d_i)).
\]
Fix a child $i$ and write $b=b_i$, $b'=M(Y_i)$, and $d'=d_i$, so that $0\le b\le b'\le M_0$ are integers and $W/4\le d'\le M_0^2W/2$. Expanding \eqref{eq:psi}, with the convention $0\log0=0$,
\[
 \Psi(b',d')-\Psi(b,d')
 =(b'-b)\log(4d'/W)+(b'-b)\log b'+b\log(b'/b).
\]
Every term is nonnegative. Moreover $\log(4d'/W)\le1+2\log M_0\le1+2L$, $\log b'\le L$, and $b\log(b'/b)\le(b'-b)/\ln2\le2(b'-b)$. Hence, since $L\ge1$,
\begin{equation}\label{eq:padding-cost}
 0\le\Psi(b',d')-\Psi(b,d')\le(3+3L)(b'-b)\le6L(b'-b).
\end{equation}
Summing over $i$ and using $\sum_i(M(Y_i)-b_i)=D_S$ proves \eqref{eq:core-padding}. The lower bound in \eqref{eq:padding-cost} also shows that, in this parameter range, $\Psi$ is nonnegative and nondecreasing in mass.
\end{proof}

\subsection{Local decomposition}\label{sec:local}

Set $\theta=1/24$. Shell masses will be bounded by $\theta/L$ times the credit they are charged to; with \cref{lem:padding-cost}, this choice ensures that padding consumes at most half of the potential decrease.

Set $\Delta=d/32$, so that sets within distance $8\Delta$ both to and from a common center have weak diameter at most $d/2$. We call a set $T\subseteq X$ \emph{compact around $s\in X$} if
\begin{equation}\label{eq:compact-target}
 \dist_H(s,v)\le2\Delta,\qquad \dist_H(v,s)\le2\Delta,
  \qquad \forall v\in T.
\end{equation}
In particular, $\wdiam_H(T)\le4\Delta$.

\begin{samepage}
\begin{lemma}[Local decomposition]\label{lem:local}
For a nonempty set $S\subseteq X$, an ordered cover $\mathcal Y=(Y_i)$ of $S$ with incidence mass $\OO(M(S))$ can be constructed in $\OO(M(S))$ time in one of the following forms:
\begin{itemize}[leftmargin=*,itemsep=3pt,topsep=5pt]
\item In the light case, every child retains $d$ and has mass at most $\rho M(S)$, for a constant $0<\rho<1$.
\item In the heavy case, the two side children retain $d$ and have mass at most $(1-\beta)M(S)$, for $\beta=1/8$. The central child receives $d/2$ and is compact around a center $s\in S$, with distances measured in $H[S]$.
\end{itemize}
The cover $\mathcal Y$ is obtained from $(S)$ by refinements, each replacing a member $A$ by a $\delta_0$-padded split of $H[A]$, where $\delta_0=\Theta(d/(L\log L))$ does not depend on $S$.

Write $d_i$ for the parameter of $Y_i$ and $D_S$ for the duplicated mass. The actual potential decrease and duplicated mass satisfy
\[
 \Pi_S\ge M(S)/16,\qquad
 D_S\le\frac{2\theta}{L}\Pi_S.
\]
\end{lemma}
\end{samepage}

Two heavy balls around the same center give the central child and smaller side children through the following split.
\begin{lemma}[Heavy split]\label{lem:heavy-split}
Let $S\subseteq X$, and let $(B_+^\circ,B_+)$ and $(B_-^\circ,B_-)$ be outward and inward $\delta$-padded balls in $H[S]$ with a common center $s$. Suppose both cores have mass at least $\beta M(S)$, for some $0<\beta<1$, and both pads have radius at most $r_{\max}$. Then the ordered cover
\begin{equation}\label{eq:heavy-children}
 (Y_1,Y_0,Y_2)=(B_+\setminus B_-^\circ,\ B_+\cap B_-,\ S\setminus B_+^\circ)
\end{equation}
of $S$ can be constructed in $\OO(M(S))$ time and satisfies the following:
\begin{enumerate}[label=(\roman*)]
\item $M(Y_1)\le(1-\beta)M(S)$ and $M(Y_2)\le(1-\beta)M(S)$;
\item every vertex of $Y_0$ has distance at most $r_{\max}$ both to and from $s$ in $H[S]$;
\item the cover is obtained from $(S)$ by two refinements, each by a $\delta$-padded split;
\item its duplicated mass satisfies $D_S\le M(\partial B_+)+M(\partial B_-)$.
\end{enumerate}
\end{lemma}
\begin{proof}
The sets cover $S$, since $(B_+\setminus B_-^\circ)\cup(B_+\cap B_-)=B_+$. The side child $Y_1$ is disjoint from $B_-^\circ$ and $Y_2$ is disjoint from $B_+^\circ$, which gives (i). Every vertex of $Y_0$ lies in both pads, which gives (ii). For (iii), refine $(S)$ into $(B_+,S\setminus B_+^\circ)$, the split given by the outward ball, and then refine $B_+$ into $(B_+\setminus B_-^\circ,B_+\cap B_-)$, the split given by the inward ball restricted to $B_+$. Restriction can only increase distances, so both splits are $\delta$-padded. For (iv), subtracting $M(S)$ from the cover's incidence mass gives
\begin{equation}\label{eq:heavy-dup}
 D_S=M(B_+\setminus B_+^\circ)+M(B_+\cap(B_-\setminus B_-^\circ))
 \le M(\partial B_+)+M(\partial B_-).
\end{equation}
All sets can be formed by linear scans.
\end{proof}

\begin{proof}[Proof of \cref{lem:local}]
Fix a current vertex set $U\subseteq S$, a center $s\in U$, and a direction $\sigma\in\{+,-\}$. We seek either a heavy ball, of mass at least $\beta M(S)=M(S)/8$, or a $\delta_0$-padded ball $(B^\circ,B)$ in $H[U]$ with a light core $M(B^\circ)<\beta M(S)$ and a small shell.

Start with the radius-zero ball around $s$, including any vertices at zero distance. If the current ball is heavy, return it. Otherwise, enlarge its radius by $\delta_0$ and inspect the newly added shell. Accept the old ball as the core $B^\circ$ and the enlarged ball as its pad $B$ if
\begin{equation}\label{eq:adaptive-shell}
 M(\partial B)\le\frac{\theta}{L}M(B^\circ)
 \lfloor\log(M(S)/M(B^\circ))\rfloor.
\end{equation}
If the test fails, make the enlarged ball the current ball and repeat. An accepted shell has mass at most $\theta M(B^\circ)\le M(B^\circ)$.

Repeated failures force the ball to become heavy. To bound their number, let $b_t$ be the mass of the current ball after $t$ failed enlargements, starting with $t=0$, and set $x_t=\log(M(S)/b_t)$. Until a heavy ball is found, $b_t<M(S)/8$, so $3<x_t\le L$ and $\lfloor x_t\rfloor\ge x_t/2$. A failed test implies
\[
 x_{t+1}<x_t-\log(1+\theta\lfloor x_t\rfloor/L)
 \le\left(1-\frac{\theta}{4L}\right)x_t.
\]
We use $\log(1+u)\ge u/2$ for $0\le u\le1$. To fit all trials within radius $\Delta$, choose
\begin{equation}\label{eq:micro-width}
 N_0=\lceil16L\lceil\log L\rceil/\theta\rceil,
 \qquad \delta_0=\frac{\Delta}{N_0}.
\end{equation}
After $N_0$ failed trials the contraction would give $x_{N_0}<3$, so the search succeeds within radius $\Delta$. An accepted core $B^\circ$ and pad $B$ satisfy $M(B)\le(1+\theta)M(B^\circ)$ by \eqref{eq:adaptive-shell}.

To obtain a decomposition, we collect light balls until their cores cover a quarter of $S$, or find a center with heavy balls in both directions. Maintain two unions $U^+,U^-$ of accepted cores. While both have mass less than $M(S)/4$, choose a vertex outside their union. Run the outward search from this vertex in $H[S\setminus U^+]$ and the inward search in $H[S\setminus U^-]$, alternating elementary operations so that the accepted core can pay for both searches. Accept the first light result into its corresponding union and discard the other search. If a search becomes heavy, pause it while the other continues. If both become heavy, use the heavy construction below.

Suppose the outward union first reaches $M(S)/4$. Retain only its accepted pads, in order, followed by the set $S\setminus U^+$. Each accepted core has mass less than $\beta M(S)$, so $U^+$ has mass in $[M(S)/4,3M(S)/8)$ and $S\setminus U^+$ has mass in $(5M(S)/8,3M(S)/4]$. Since $\theta<1$, every pad and $S\setminus U^+$ have mass at most $\rho M(S)$, for example with $\rho=7/8$. The two unions may overlap; only the balls of the winning orientation are retained. If the inward union wins, the child order is reversed. Each accepted ball refines the entry consisting of the vertices not yet in earlier cores by a $\delta_0$-padded split.

If both searches become heavy, then the in-ball and out-ball of radius $\Delta$ around the chosen vertex $s$ in $H[S]$ have mass at least $\beta M(S)$, since balls in an induced subgraph are contained in the corresponding balls of $H[S]$. Set
\[
 N_h=\lceil16L/\theta\rceil,\qquad \delta_h=\Delta/N_h.
\]
Partition $[\Delta,2\Delta]$ into $N_h$ shells of width $\delta_h$ in each direction. These shells are disjoint within a direction and have total mass at most $M(S)$. At least one has mass at most $M(S)/N_h\le\theta M(S)/(16L)$. Choose such a shell in each direction and write the outward and inward $\delta_h$-padded balls as $(B_+^\circ,B_+)$ and $(B_-^\circ,B_-)$, respectively. The pad radii are at most $2\Delta$.

Apply \cref{lem:heavy-split} to these two balls. Its side children have mass at most $(1-\beta)M(S)$, and its central child is compact around $s$, with weak diameter at most $4\Delta<d/2$. The central child receives parameter $d/2$ and the side children retain $d$. The split width is at least $\delta_0$, since $\delta_h\ge\delta_0$.

In both cases, every refinement is a $\delta_0$-padded split. The light cover has incidence mass at most $(1+\theta)M(S)$, since its cores and $S\setminus U^\pm$ partition $S$. The heavy cover consists of three subsets of $S$ and has incidence mass at most $3M(S)$.

Monotone Dijkstra exploration with lazy initialization takes time linear in the explored mass. Shell tests add no larger cost: a failed shell contains a new vertex, and an empty shell succeeds. When a light search finishes, the discarded competing search has used at most the same number of operations. The accepted pad has mass at most $(1+\theta)$ times its core mass, so that core pays for both searches. Cores are disjoint within each orientation, giving total work $\OO(M(S))$. A pair of heavy searches and the subsequent scans for the two heavy shells also cost $\OO(M(S))$.

For the light output, take the core partition consisting of the accepted cores and $S\setminus U^\pm$. For the heavy output, choose any core partition of the three children. Let $\Pi^\circ_S$ be the potential decrease before padding, given by \eqref{eq:core-credit}. Every child retaining $d$ has mass at most $7M(S)/8$ in both cases, so its core contributes at least $b_i\log(8/7)\ge b_i/8$ to the size credit. A core whose child receives $d/2$ contributes at least $b_i$ to the diameter credit. Since the core masses sum to $M(S)$, we have $\Pi^\circ_S\ge M(S)/8$. In the light case, the shell tests give
\[
 D_S\le\frac{\theta}{L}\sum_i b_i\log(M(S)/b_i)
 \le\frac{\theta}{L}\Pi^\circ_S.
\]
In the heavy case, \cref{lem:heavy-split} bounds $D_S$ by the sum of the two shell masses, so $D_S\le\theta M(S)/(8L)\le\theta\Pi^\circ_S/L$. Both cases therefore satisfy
\begin{equation}\label{eq:local-dup}
 D_S\le\frac{\theta}{L}\Pi^\circ_S,\qquad \Pi^\circ_S\ge M(S)/8.
\end{equation}
By \cref{lem:padding-cost}, padding consumes $\Pi^\circ_S-\Pi_S\le6LD_S\le6\theta\Pi^\circ_S=\Pi^\circ_S/4$, so
\begin{equation}\label{eq:local-drop}
 \Pi_S\ge\Pi^\circ_S/2.
\end{equation}
Consequently $\Pi_S\ge M(S)/16$ and $D_S\le2\theta\Pi_S/L$.
\end{proof}

\subsection{Residual cover}\label{sec:residual-cover}

Given a set $R\subseteq X$, we seek either a cover of $R$ by pieces of mass at most a cutoff $Z$, or a subset of mass greater than $Z$ that is compact around a center. Let $\tau$ be the largest power of two at most $(\log L)/2$, so that
\begin{equation}\label{eq:tau-range}
 \tfrac14\log L<\tau\le\tfrac12\log L,
\end{equation}
and assume that $L$ is large enough that $\tau\ge4$. Set
\begin{equation}\label{eq:residual-cutoff}
 Z=M/2^\tau.
\end{equation}

\begin{samepage}
\begin{lemma}[Residual cover or center]\label{lem:residual}
For any set $R\subseteq X$, there is an algorithm $\Residual(R)$ that takes $\OO(M(R)\tau)$ time and returns one of the following:
\begin{itemize}[leftmargin=*,itemsep=3pt,topsep=5pt]
\item A set $T\subseteq R$ of mass greater than $Z$, compact around a center $s\in R$.
\item An ordered cover $\mathcal Y_{\res}=(Y_t)$ of $R$ whose members have mass at most $Z$, with child parameters $d_t$ and total incidence mass $\OO(M(R))$. Each child retains parameter $d$, or receives $d/2$ and has weak diameter at most $d/2$ in $H$. Every walk $P$ in $H[R]$ satisfies
\begin{equation}\label{eq:residual-ascents}
 \ell(P)\le d\quad\Longrightarrow\quad a_{\mathcal Y_{\res}}(P)=\OO(L\tau).
\end{equation}
The duplicated mass $D_{\res}$ and actual potential decrease $\Pi_{\res}:=\Psi(M(R),d)-\sum_t\Psi(M(Y_t),d_t)$ satisfy
\begin{equation}\label{eq:residual-dup}
 D_{\res}\le\frac{2\theta}{L}\Pi_{\res}.
\end{equation}
\end{itemize}
\end{lemma}
\end{samepage}

We build this cover by repeatedly dividing members of mass greater than $Z$ using \cref{lem:local}. Each division produces children of smaller mass, except possibly a central child of small weak diameter; if that child still has mass greater than $Z$, it supplies the set and center in the first outcome.

To track duplication across repeated divisions, consider a sequence of refinements from $(S,d)$ to child instances $(Y_t,d_t)$. Let $D_A$ and $\Pi_A$ be the duplicated mass and actual potential decrease of each replaced occurrence $A$ with respect to its replacement cover. Only that occurrence and its replacement children change the cover's mass and potential, so
\begin{align}
 \sum_t M(Y_t)-M(S)&=\sum_{A\text{ replaced}}D_A,\nonumber\\
 \Psi(M(S),d)-\sum_t\Psi(M(Y_t),d_t)
 &=\sum_{A\text{ replaced}}\Pi_A.\label{eq:refinement-accounting}
\end{align}
Occurrences are counted separately even when they have the same vertex set.

\begin{proof}[Proof of \cref{lem:residual}]
If $R$ is empty, return the empty cover. Otherwise start with an ordered list containing just $(R,d)$. Refine any member of mass greater than $Z$ using \cref{lem:local}, discarding empty children.

If a heavy case produces a central child of mass greater than $Z$, return that child and its center immediately. Central children of mass at most $Z$ receive parameter $d/2$; all other children retain $d$. Continue replacing oversized members in depth-first order. If no center is returned, the resulting list is an ordered cover whose members all have mass at most $Z$.

To bound the work even when a center is returned early, consider the completed refinement, which continues refining the other children but leaves every central child unexpanded. The actual execution performs only a prefix of this construction. The leaves of its refinement tree are central children and sets of mass at most $Z$. Every expanded child has mass at most $\rho$ times its parent's mass, and expanded masses exceed $Z$, so the depth is $\OO(\log(M/Z))=\OO(\tau)$. All expanded nodes retain the same $d$ and use the same $\Delta$. By \cref{lem:local}, each division consists of refinements by $\delta_0$-padded splits, with the common width $\delta_0$. Performing the divisions in sequence therefore obtains the ordered leaf cover from $(R)$ by $\delta_0$-padded splits, and \cref{lem:composition} bounds its ascent number by $\lfloor\ell(P)/\delta_0\rfloor$. For $\ell(P)\le d$, this is $\OO(L\log L)=\OO(L\tau)$, proving \eqref{eq:residual-ascents}. Every central child is compact in an induced subgraph of $H$, hence in $H$ itself, and has weak diameter at most $4\Delta<d/2$. A leaf above the cutoff supplies the required set and center.

A linear incidence bound for each division does not imply a linear bound for the final cover. We therefore use the accounting guarantee of \cref{lem:local} to charge each additional copy to the actual potential decrease of the division that creates it.

Let $Y_t$, with parameters $d_t$, be the leaves of the completed refinement, with duplicated mass $D_{\res}$ and actual potential decrease $\Pi_{\res}$. For each internal occurrence $A$, write $D_A$ and $\Pi_A$ for its duplicated mass and actual potential decrease. \Cref{lem:local} and the refinement identities \eqref{eq:refinement-accounting} give
\[
 D_{\res}=\sum_{A\text{ internal}}D_A
 \le\frac{2\theta}{L}\sum_{A\text{ internal}}\Pi_A
 =\frac{2\theta}{L}\Pi_{\res}.
\]
This proves \eqref{eq:residual-dup} whenever the routine returns a cover, and gives the same bound for the completed refinement used to analyze an early return.

Choose any core partition of the completed cover, with masses $b_t$. Since $\Psi$ is nondecreasing in mass, $\Pi_{\res}\le\Pi^\circ_R$, and \eqref{eq:core-credit} gives
\[
 \Pi_{\res}\le\sum_{t:b_t>0}b_t\log(M(R)/b_t)+\sum_{t:d_t=d/2}b_t
 \le M(R)\log M(R)+M(R)=\OO(M(R)L),
\]
since positive $b_t$ are at least one. Hence $D_{\res}=\OO(M(R))$ and the final incidence mass is $M(R)+D_{\res}=\OO(M(R))$. At any depth, the total mass of the occurrences is at most this final incidence mass, since their descendant leaves cover each occurrence separately. Linear work per occurrence over $\OO(\tau)$ levels gives $\OO(M(R)\tau)$ total work, which also bounds an early return.
\end{proof}

\subsection{Peeling}\label{sec:peeling}

We peel around the stored centers so that any large compact set left in the residual is disjoint from the compact sets that supplied those centers. The peeled children must be light and have disjoint cores, and their solutions must be cheap to merge once the residual has been solved. We charge a core of mass $b$ its size credit $b\log(M/b)$; this charge also bounds the mass duplicated by padding.

We write $\calS$ for the list of stored pairs $(s,\sigma)$ and $n_X=|X|$. Each stored center $s\in X$ comes with a direction $\sigma\in\{+,-\}$ such that $M(B_H^\sigma(s,7\Delta))<\beta M=M/8$.

$\Peel$ keeps an active set $U$, initially $X$. It visits test radii from largest to smallest, and at each radius scans the stored centers in a fixed order. If a center's ball has too little active mass, it skips the center. Otherwise it grows the ball until it finds a small shell, saves the pad as a peeled child, and removes the core from $U$. Every ball uses distances in $H$ and contains only vertices still in $U$.

$\Peel$ runs in stages $j=1,\ldots,J$. Stage $j$ uses mass threshold $M/2^{t_j}$ and test radius $\gamma_j$, where
\begin{equation}\label{eq:stages}
 t_j=2^{j+1}\quad(0\le j\le J),\qquad J=\log\tau-1,
\end{equation}
so that $t_0=2$ and $t_J=\tau$. The thresholds decrease to $M/2^\tau=Z$, and the test radii decrease from $7\Delta$ toward $6\Delta$. A ball tested at radius $\gamma_j$ may grow up to the preceding radius $\gamma_{j-1}$; we divide this interval into $N$ shells of width $\delta_j$ by setting
\begin{equation}\label{eq:stage-radii}
 N=\lceil4L/\theta\rceil,\qquad
 \gamma_j=6\Delta+\Delta/2^j,\qquad
 \delta_j=\frac{\Delta}{N2^j}=\frac{2\Delta}{Nt_j}.
\end{equation}
Thus $\gamma_0=7\Delta$, $\gamma_J>6\Delta$, and $\delta_j=\Theta(d/(Lt_j))$.

\begin{samepage}
\begin{quote}\small
\textbf{Algorithm 1.} $\Peel$ on the stored pairs $\calS$.
\begin{enumerate}[label=\arabic*.]
\item Set $U=X$ and initialize an empty list of accepted balls.
\item For $j=1,\ldots,J$ (so that the test radius $\gamma_j$ decreases), scan every stored pair $(s,\sigma)\in\calS$ in a fixed order:
\begin{enumerate}[label=(\alph*),leftmargin=*,itemsep=3pt,topsep=3pt]
\item If $M(B_H^\sigma(s,\gamma_j)\cap U)<M/2^{t_j}$, skip this pair.
\item Otherwise start at $r=\gamma_j$ and form
\begin{equation}\label{eq:filtered}
 B^\circ=B_H^\sigma(s,r)\cap U,\qquad
 B=B_H^\sigma(s,r+\delta_j)\cap U.
\end{equation}
If $M(\partial B)>\theta t_jM(B^\circ)/(2L)$, increase $r$ by $\delta_j$ and repeat. On success, record $(B^\circ,B)$, $\sigma$, and $j$, and set $U\leftarrow U\setminus B^\circ$.
\end{enumerate}
\item Return the accepted balls and the residual $R=U$.
\end{enumerate}
\end{quote}
\end{samepage}

\begin{lemma}[Peeling]\label{lem:peeling}
$\Peel$ constructs at most $|\calS|$ peeled children $Y_i$ with nonempty cores $Y_i^\circ\subseteq Y_i$ and a residual $R$, satisfying
\[
 X=R\sqcup\bigsqcup_i Y_i^\circ,\qquad M(Y_i)<\beta M.
\]
If $T\subseteq R$ is compact around a center in $R$ and $M(T)>Z$, then $T$ is disjoint from every set compact around a stored center. Writing $b_i=M(Y_i^\circ)$, the total duplicated mass satisfies
\begin{equation}\label{eq:peeling-dup}
 \sum_i M(\partial Y_i)
 \le\frac{\theta}{L}\sum_i b_i\log(M/b_i).
\end{equation}
Given each stored ball $B_H^\sigma(s,7\Delta)$ with its vertices sorted by distance, the construction takes $\OO(n_X(|\calS|+1)\log\tau)$ time.
\end{lemma}

\begin{lemma}[Peeling reconstruction]\label{lem:peeling-reconstruction}
For the peeled children and residual constructed by \cref{lem:peeling}, given valid $c$-potentials on every $G[Y_i]$ and on $G[R]$, there is a deterministic algorithm that computes a valid $c$-potential on $G[X]$, finds a negative $c$-cycle, or returns a walk satisfying \eqref{eq:certificate}, in time
\begin{equation}\label{eq:peeling-merge}
 \OO\!\left(L\sum_i b_i\log(M/b_i)+L\tau M(R)\right).
\end{equation}
\end{lemma}

\begin{proof}[Proof of \cref{lem:peeling}]
Each acceptance removes its core from $U$, so the cores are disjoint and, together with the residual, partition $X$.

We prove by induction that, after stage $j$, the active set $U$ satisfies
\begin{equation}\label{eq:stage-invariant}
 M(B_H^\sigma(s,\gamma_j)\cap U)\le M/2^{t_j}
\end{equation}
for every stored pair $(s,\sigma)$. For $j=0$, we have $U=X$ and $\gamma_0=7\Delta$. The light-ball assumption gives mass less than $M/8$, hence at most $M/2^{t_0}=M/4$.

Assume the bound holds after stage $j-1$. Deletions during stage $j$ preserve that bound at radius $\gamma_{j-1}$. Consider a stored pair when the scan reaches it. If its test ball at $\gamma_j$ has mass below $M/2^{t_j}$, the pair is skipped and already satisfies the desired bound at radius $\gamma_j$.

Otherwise the test starts a shell search. Keep $U$ fixed during this search, and let $b_0$ and $b_N$ be the active ball masses at the two ends $\gamma_j$ and $\gamma_{j-1}$ of the search interval. If all $N$ shell tests failed, then
\[
 b_N>b_0\left(1+\frac{\theta t_j}{2L}\right)^N.
\]
Since $\theta t_j/(2L)\le1$ and $b_0\ge M/2^{t_j}$,
\[
 \log(b_N/b_0)>\frac{N\theta t_j}{4L}\ge t_j,
\]
which would give $b_N>M$. Thus the search succeeds before leaving the interval. Its accepted core contains every active vertex in the test ball at $\gamma_j$, so removing that core leaves zero mass in this ball, and the pair satisfies the desired bound. Since the radii decrease, all later test balls of this pair are empty; it therefore produces at most one ball, although it continues to be scanned.

Once a pair has been processed, later deletions preserve its bound. After scanning all pairs, \eqref{eq:stage-invariant} therefore holds for $j$, completing the induction. At $j=J$, the bound is $Z$ and $\gamma_J>6\Delta$. Hence, for every stored pair $(s_i,\sigma_i)$,
\begin{equation}\label{eq:all-centers}
 M(B_H^{\sigma_i}(s_i,6\Delta)\cap R)\le Z.
\end{equation}
Let $T\subseteq R$ be compact around $s\in R$, and suppose that it intersects a set $T_i$ compact around a stored center $s_i$ at a vertex $v$. For every $u\in T$, compactness gives
\begin{align*}
 \dist_H(s_i,u)&\le\dist_H(s_i,v)+\dist_H(v,s)+\dist_H(s,u)
 \le6\Delta,\\
 \dist_H(u,s_i)&\le\dist_H(u,s)+\dist_H(s,v)+\dist_H(v,s_i)
 \le6\Delta.
\end{align*}
Thus $T\subseteq B_H^{\sigma_i}(s_i,6\Delta)\cap R$, so \eqref{eq:all-centers} gives $M(T)\le Z$. This proves the compact-set guarantee.

For a core $Y_i^\circ$ accepted during stage $j$, write $b_i=M(Y_i^\circ)$. The test at $\gamma_j$ gives $b_i\ge M/2^{t_j}$. The bound after stage $j-1$ applies throughout the search interval, giving $b_i\le M/2^{t_{j-1}}$. Since $t_{j-1}=t_j/2$, the core and its accepted shell satisfy
\begin{equation}\label{eq:rank-bracket}
 \frac{t_j}{2}\le\log(M/b_i)\le t_j,\qquad
 M(\partial Y_i)\le\frac{\theta}{L}b_i\log(M/b_i).
\end{equation}
For the shell bound, the successful test gives $M(\partial Y_i)\le\theta t_jb_i/(2L)$, and $t_j/2\le\log(M/b_i)$. Summing over accepted cores proves \eqref{eq:peeling-dup}. Each pad lies in $B_H^\sigma(s,7\Delta)$, so $M(Y_i)<\beta M$.

For the work bound, keep the vertices and distances of each stored ball $B_H^\sigma(s,7\Delta)$ in nondecreasing distance order. Each size test scans at most $n_X$ entries and checks whether each vertex remains in $U$. During a shell search, two monotone pointers maintain core and pad mass in this list. Every failed test has positive shell mass and moves at least one active vertex into the next core; an empty shell succeeds. Hence pointer advances and trial processing cost $\OO(n_X)$ per accepted ball. Recording its vertex lists and deleting its core has the same bound. There are $|\calS|J$ size tests and at most $|\calS|$ accepted balls. Initializing $U$ costs $\OO(n_X)$, yielding $\OO(n_XJ(|\calS|+1))=\OO(n_X(|\calS|+1)\log\tau)$ work.
\end{proof}

\begin{proof}[Proof of \cref{lem:peeling-reconstruction}]
Start with the supplied valid potential on $G[R]$. Let $U_j$ be the active set before stage $j$, and let $D_j$ be the total shell mass accepted during that stage. Set $U_{J+1}=R$. Process the stages in reverse order, from $J$ down to $1$. When processing stage $j$, a valid potential on $G[U_{j+1}]$ is already available. We combine it with the given potentials on the children peeled during stage $j$ to obtain a potential on $G[U_j]$.

Starting from $(U_j)$, each outward peel with core $B^\circ$ and pad $B$ refines the entry equal to the current active set $U$ into $(B,U\setminus B^\circ)$, and each inward peel refines it into $(U\setminus B^\circ,B)$. These refinements give an ordered cover $\mathcal Y_j$ consisting of the outward pads in chronological order, the solved child $U_{j+1}$, and the inward pads in reverse chronological order. Each split is $\delta_j$-padded in $H[U]$: before intersecting with $U$, it comes from a padded ball in $H$, and restriction can only increase distances. By \cref{lem:composition}, every walk $P$ in $G[U_j]$ with $\ell(P)\le d$ satisfies $a_{\mathcal Y_j}(P)\le\lfloor d/\delta_j\rfloor=\OO(Lt_j)$. The solved set $U_{j+1}$ enters as one induced child with a valid potential. The stage's peeled cores and $U_{j+1}$ form a core partition of $\mathcal Y_j$, with duplicated mass $D_j$. Hence
\begin{equation}\label{eq:stage-incidence}
 \sum_{Y\in\mathcal Y_j}M(Y)=M(U_j)+D_j.
\end{equation}
By \cref{thm:merge}, this merge takes $\OO(Lt_j(M(U_j)+D_j))$ time. If a stage accepted no ball, the potential on $U_{j+1}$ already serves for $U_j$.

Consider a core $Y_i^\circ$ accepted during stage $j$. It is contained in $U_h$ exactly for $1\le h\le j$. Since $t_{h+1}=2t_h$, its contribution to $\sum_h t_hM(U_h)$ is at most
\[
 b_i\sum_{h=1}^j t_h<2b_it_j\le4b_i\log(M/b_i),
\]
where the last inequality is \eqref{eq:rank-bracket}. The vertices of $R$ lie in every $U_j$, and $\sum_{j=1}^Jt_j<2\tau$. Summing over the core partition of $X$ gives
\begin{equation}\label{eq:suffix-sum}
 \sum_{j=1}^Jt_jM(U_j)
 \le4\sum_i b_i\log(M/b_i)+2\tau M(R).
\end{equation}
For the shell terms, \eqref{eq:peeling-dup} gives
\[
 \sum_{j=1}^Jt_jD_j
 \le\tau\sum_i M(\partial Y_i)
 \le\frac{\theta\tau}{L}\sum_i b_i\log(M/b_i)
 \le\theta\sum_i b_i\log(M/b_i),
\]
since $\tau\le L$. These bounds prove \eqref{eq:peeling-merge}. Every intermediate parent is an induced subset of $X$, and all merges use the same parameter $d$. Each successful merge supplies a valid potential for the next one. A failed merge returns a negative $c$-cycle or a walk $P$ in $G[X]$ with $c(P)<0$ and $\ell(P)>d$.
\end{proof}

\subsection{Construction}\label{sec:construction}

We alternate peeling and residual division. After each unsuccessful attempt, we store a new center with a light direction, and the next attempt starts again from $X$. All choices use fixed deterministic tie breaking.

\begin{samepage}
\begin{quote}\small
\textbf{Algorithm 2.} $\Decompose(X,d)$.
\begin{enumerate}[label=\arabic*.]
\item Start with an empty list $\calS$ of pairs $(s,\sigma)$.
\item Run $\Peel$ from all of $X$ using $\calS$. Let $R$ be the residual and keep the peeled balls of this attempt temporarily.
\item Run $\Residual(R)$. If it returns a cover, return its children together with the peeled children.
\item Otherwise let $(T,s)$ be the returned compact set and center. Compute both radius-$7\Delta$ balls around $s$ in $H$. If both have mass at least $\beta M$, return the three-child decomposition of $X$ described below.
\item Append $(s,\sigma)$ to $\calS$ for a light direction $\sigma$. Discard the peeled balls and the residual construction, and return to step 2.
\end{enumerate}
\end{quote}
\end{samepage}

When both radius-$7\Delta$ balls around $s$ are heavy, use $N_h=\lceil16L/\theta\rceil$ shells in $[7\Delta,8\Delta]$ to find an inward and an outward padded ball whose shells each have mass at most $\theta M/(16L)$. Apply \cref{lem:heavy-split} with $S=X$, $r_{\max}=8\Delta$, and $\delta_h=\Delta/N_h=\Theta(d/L)$. The side children have mass at most $(1-\beta)M$, the central child has weak diameter at most $16\Delta=d/2$, and duplication is at most $\theta M/(8L)$. The central child receives parameter $d/2$ and the side children retain $d$. The shell searches and split construction take $\OO(M)$ time.

\begin{lemma}[Final guarantees]\label{lem:final-guarantees}
If Algorithm 2 returns peeled children and a residual cover, the children have total incidence mass $\OO(M)$ and each has mass at most $M/8$. Every child assigned $d/2$ has weak diameter at most $d/2$ in $H$.
\end{lemma}
\begin{proof}
The peeled cores $Y_i^\circ$ are disjoint. Write $b_i=M(Y_i^\circ)$. Since $b_i\ge1$ and $\log(M/b_i)\le L$, the duplication bound in \cref{lem:peeling} gives
\[
 \sum_i M(Y_i)\le(1+\theta)\sum_i b_i\le(1+\theta)M.
\]
The residual $R$ has a cover of incidence mass $\OO(M(R))$ by \cref{lem:residual}. Each peeled pad has mass less than $M/8$, and each residual child has mass at most $Z=M/2^\tau\le M/16$. \Cref{lem:residual} also gives the diameter bound for every child assigned $d/2$.
\end{proof}

To bound the number of attempts, note that $X$ contains fewer than $2^\tau=M/Z$ pairwise disjoint sets of mass greater than $Z$.

\begin{lemma}[Construction work]\label{lem:construction-work}
Algorithm 2 terminates after at most $2^\tau$ attempts and takes time
\begin{equation}\label{eq:construction-work}
 \OO(M2^\tau\tau+n_X2^{2\tau}\log\tau),
\end{equation}
including every unsuccessful attempt.
\end{lemma}
\begin{proof}
Associate each stored pair $(s_i,\sigma_i)$ with the compact set $T_i$ that exposed its center. These sets are pairwise disjoint by \cref{lem:peeling}: each new proposal is a compact set $T\subseteq R$ of mass greater than $Z$, so it is disjoint from every $T_i$ associated with a stored center.

If $|\calS|$ pairs have been stored, disjointness gives $|\calS|Z<\sum_iM(T_i)\le M$, so $|\calS|<2^\tau$. Each attempt either returns a decomposition or stores a new pair, so the algorithm terminates within $2^\tau$ attempts.

For each proposed center, the two radius-$7\Delta$ balls and their vertices in distance order are computed once by truncated searches in $H$, taking $\OO(M)$ time. Including a possible final heavy test, this is $\OO(M2^\tau)$. Each residual construction costs at most $\OO(M\tau)$, giving $\OO(M2^\tau\tau)$ overall.

The peeling attempts use $0,1,\ldots,|\calS|$ stored pairs. Their total work is
\[
 \OO\!\left(n_X\log\tau\sum_{k=0}^{|\calS|}(k+1)\right)
 =\OO(n_X(|\calS|+1)^2\log\tau).
\]
Induced edge lists are created only for the returned children. Their construction costs their incidence mass, which is $\OO(M)$ by \cref{lem:final-guarantees} when returning peeled children and a residual cover, and by the heavy construction otherwise.
\end{proof}

The quadratic cost of rescanning the stored balls is what limits the cutoff. By \eqref{eq:tau-range},
\begin{equation}\label{eq:cutoff-bounds}
 L^{1/4}<2^\tau\le\sqrt L,\qquad 2^{2\tau}\le L,\qquad 2^\tau\tau=\OO(L).
\end{equation}
Since $n_X\lambda\le M$ and $\log\tau=\OO(\lambda)$, every term in \eqref{eq:construction-work} is $\OO(ML)$. When $L$ is too small for $\tau\ge4$, we apply \cref{lem:local} directly to $X$; this happens only when $L$ is bounded by a constant.

\begin{proof}[Proof of \cref{thm:decomp}]
The work bound follows from \cref{lem:construction-work} with this choice of $\tau$. For peeled children and a residual cover, \cref{lem:final-guarantees} gives the incidence and child-parameter guarantees. It remains to show that padding leaves enough potential decrease to pay for this work.

We index peeled children by $i$ and residual children by $t$. Take a core partition $(Y_t^\circ)$ of the residual cover $\mathcal Y_{\res}=(Y_t)$, with masses $b_t=M(Y_t^\circ)$. Together with the peeled cores $Y_i^\circ$, of masses $b_i$, these form a core partition of the final cover of $X$. Its size credit splits as
\begin{equation}\label{eq:size-credit-split}
 \kappa=\kappa_{\peel}+\kappa_{\res},\qquad
 \kappa_{\peel}=\sum_i b_i\log(M/b_i),\qquad
 \kappa_{\res}=\sum_{t:b_t>0} b_t\log(M/b_t).
\end{equation}
Let $b_{1/2}$ be the total core mass of children with parameter $d/2$. By \eqref{eq:core-credit}, the potential decrease before padding is
\begin{equation}\label{eq:accounting-credit}
 \Pi^\circ_X=\kappa+b_{1/2}.
\end{equation}
Padding introduces $D_X=D_{\peel}+D_{\res}$ extra mass, where
\[
 D_{\peel}=\sum_iM(\partial Y_i),\qquad
 D_{\res}=\sum_tM(Y_t)-M(R).
\]

\Cref{lem:peeling} gives $D_{\peel}\le\theta\kappa_{\peel}/L$. For the residual cover, \eqref{eq:residual-dup} gives $D_{\res}\le2\theta\Pi_{\res}/L$. The actual decrease $\Pi_{\res}$ is no larger than the decrease before padding for the chosen core partition, so by \eqref{eq:core-credit},
\[
 D_{\res}\le\frac{2\theta}{L}
       \left(\sum_{t:b_t>0} b_t\log(M(R)/b_t)+b_{1/2}\right).
\]
When $M(R)>0$, the identity
\begin{equation}\label{eq:credit-split}
 \kappa_{\res}=M(R)\log(M/M(R))+\sum_{t:b_t>0} b_t\log(M(R)/b_t)
\end{equation}
shows that the sum in this bound is at most $\kappa_{\res}$; when $M(R)=0$, we have $D_{\res}=0$. We conclude
\[
 D_X=D_{\peel}+D_{\res}
 \le\frac{2\theta}{L}(\kappa_{\peel}+\kappa_{\res}+b_{1/2})
 =\frac{2\theta}{L}\Pi^\circ_X.
\]

Peeled pads have mass less than $M/8$ and residual children at most $M/2^\tau\le M/16$. Thus every positive core mass is at most $M/8$, and $\kappa\ge M\log8$. Positive core masses are at least one, so $\kappa\le M\log M$ and $b_{1/2}\le M$. We have proved
\begin{equation}\label{eq:final-estimates}
 D_X\le\frac{2\theta}{L}\Pi^\circ_X,\qquad \Pi^\circ_X=\Omega(M),\qquad
 \Pi^\circ_X=\OO(ML).
\end{equation}
By \cref{lem:padding-cost}, $\Pi^\circ_X-\Pi_X\le6LD_X\le12\theta\Pi^\circ_X=\Pi^\circ_X/2$, so
\begin{equation}\label{eq:progress}
 \Pi_X\ge\Pi^\circ_X/2=\Omega(M).
\end{equation}

For the three-child decomposition constructed directly on $X$, choose any core partition. As in the proof of \cref{lem:local}, the side children have mass at most $7M/8$ and the central child receives $d/2$, so $\Pi^\circ_X\ge M/8$. \Cref{lem:heavy-split} bounds duplication by $\theta M/(8L)\le\theta\Pi^\circ_X/L$, so the incidence mass is $\OO(M)$, and \cref{lem:padding-cost} gives $\Pi_X\ge\Pi^\circ_X-6\theta\Pi^\circ_X\ge M/16$. The side children and central child have the required mass and diameter bounds.

If $\tau<4$, use \cref{lem:local} directly on $X$. It takes $\OO(M)$ time and gives the required child bounds and actual potential decrease $\Omega(M)$.
\end{proof}

\subsection{Reconstruction}\label{sec:reconstruction}

\begin{proof}[Proof of \cref{thm:reconstruction}]
Suppose the construction returns peeled children and a residual cover $\mathcal Y_{\res}$ of $R$. Valid potentials on all these children have been supplied. We first merge $\mathcal Y_{\res}$ to obtain a potential on $G[R]$, then combine it with the potentials on the peeled children using \cref{lem:peeling-reconstruction}.

By \cref{lem:residual}, the cover $\mathcal Y_{\res}$ has incidence mass $\OO(M(R))$, and every walk in $G[R]$ of clipped length at most $d$ has ascent number $\OO(L\tau)$. \Cref{thm:merge} therefore merges this cover in $\OO(L\tau M(R))$ time. Applying \cref{lem:peeling-reconstruction} gives total work $\OO(L(\kappa_{\peel}+\tau M(R)))$, where $\kappa_{\peel}$ is the size credit of the peeled cores in \eqref{eq:size-credit-split}.

The small residual children pay for the $\tau M(R)$ term. Use the core partition from the proof of \cref{thm:decomp}. Each positive residual core mass satisfies $b_t\le M(Y_t)\le Z=M/2^\tau$, so
\begin{equation}\label{eq:residual-credit}
 \kappa_{\res}=\sum_{t:b_t>0}b_t\log(M/b_t)\ge\tau M(R).
\end{equation}
Together with \eqref{eq:accounting-credit} and \eqref{eq:progress}, this gives
\begin{equation}\label{eq:combine-charge}
 L(\kappa_{\peel}+\tau M(R))\le L(\kappa_{\peel}+\kappa_{\res})
 =L\kappa\le2L\Pi_X.
\end{equation}

For the three-child decomposition constructed directly on $X$, \cref{lem:heavy-split,lem:composition} with $\delta_h=\Theta(d/L)$ bound the ascent number by $\OO(L)$ on walks of clipped length at most $d$. Its incidence mass is $\OO(M)$, so \cref{thm:merge} takes $\OO(ML)$ time. This is $\OO(L\Pi_X)$ because $\Pi_X=\Omega(M)$ by \cref{thm:decomp}.

In the remaining case $\tau<4$, the cutoff choice bounds $L$ by a constant. The cover from \cref{lem:local} has incidence mass $\OO(M)$, and by \cref{lem:composition} its ascent number is $\OO(1)$ on walks of clipped length at most $d$. Its merge takes $\OO(M)$ time, again within $\OO(L\Pi_X)$.

All merges use parameter $d$ on induced subsets of $X$. If any merge returns a negative $c$-cycle or a walk $P$ with $c(P)<0$ and $\ell(P)>d$, that cycle or walk lies in $G[X]$. Otherwise the merges produce the required valid potential on $G[X]$.
\end{proof}

\section{Shortest paths}\label{sec:sssp}

\begin{theorem}[One scaling step]\label{thm:one-scale}
There is a deterministic algorithm that, given a graph with integral edge weights at least $-W$, finds a negative-weight cycle or a valid potential for $c=w+W/2$ in $\OO(M_0L^2)$ time.
\end{theorem}

\subsection{Algorithm}

The initial instance is $(V,d_0)$, where
\begin{equation}\label{eq:root-parameter}
 d_0=M_0^2W/2.
\end{equation}
Every descendant parameter is obtained by retaining or halving its parent's parameter. We maintain the invariant
\begin{equation}\label{eq:weak-invariant}
 d<d_0\quad\Longrightarrow\quad\wdiam_{\bar G}(X)\le d.
\end{equation}

\begin{samepage}
\begin{quote}\small
\textbf{Algorithm 3.} $\Solve(X,d)$.
\begin{enumerate}[label=\arabic*.]
\item If $X$ is empty, return the empty potential. If $d<W/2$, scan $G[X]$. Return the zero potential if all its $c$-weights are nonnegative; otherwise use a negative $c$-edge to return a negative $w$-cycle as described below.
\item Construct the children using \cref{thm:decomp}.
\item Recursively solve the children. If a call returns a negative $w$-cycle, return it immediately.
\item Merge the child potentials using \cref{thm:reconstruction}. Return the parent potential if successful. Convert a returned cycle or walk certificate to a negative $w$-cycle and return it.
\end{enumerate}
\end{quote}
\end{samepage}

In the size base case, a graph with few enough vertices can be solved directly with Bellman--Ford.

\subsection{Correctness}

If a child retains $d$, it is a subset of its parent and inherits \eqref{eq:weak-invariant} whenever that invariant applies. If it receives $d/2$, \cref{thm:decomp} certifies weak diameter at most $d/2$ in the parent's clipped graph $H$, and therefore in $\bar G$. This establishes the invariant at the first diameter decrease and preserves it thereafter.

\begin{lemma}[Converting a certificate]\label{lem:cycle-certificate}
In a recursive instance satisfying \eqref{eq:weak-invariant}, every failure certificate from \cref{thm:reconstruction} yields a negative-weight cycle in $G$ under $w$. Apart from reading the certificate, at most one $\OO(M_0)$-time shortest-path search in $\bar G$ is needed.
\end{lemma}
\begin{proof}
A negative $c$-cycle is also negative for $w$, since $w(e)\le c(e)$. Otherwise let $P$ be the returned walk with $c(P)<0$ and $\ell(P)>d$. On a negative $c$-edge we have $c(e)\ge-W/2$ and $w(e)=c(e)-W/2\le2c(e)$; on the other edges $w(e)\le c(e)$. Summing these bounds, the nonnegative $c$-edges contribute $\ell(P)$ and the negative ones contribute at most $2(c(P)-\ell(P))$. Thus
\begin{equation}\label{eq:negative-original}
 w(P)\le2c(P)-\ell(P)<-\ell(P)<-d.
\end{equation}

For $d<d_0$, the invariant \eqref{eq:weak-invariant} gives a path in $\bar G$ from the end of $P$ back to its start, of clipped length at most $d$. We compute it by Dijkstra. Its $w$-weight is at most its clipped length, so it closes $P$ into a negative closed walk, and cycle erasure extracts a negative cycle.

For $d=d_0$, that is, at the root and at descendants that retain $d_0$, cycle-erase $P$ directly. If no removed cycle has negative $w$-weight, the remaining simple path has weight less than $-d_0$ by \eqref{eq:negative-original}. A simple path has at most $n-1$ edges, each of weight at least $-W$, whereas
\[
 d_0=M_0^2W/2\ge n^2W/2>(n-1)W.
\]
This is impossible, so some removed cycle is negative.
\end{proof}

At the base case $d<W/2$, a negative $c$-edge $(u,v)$ has $w(u,v)<-W/2$. Since $d<d_0$, a clipped return path from $v$ to $u$ has length at most $d<W/2$ and closes it to a negative $w$-cycle. Without such an edge, the zero potential is valid for $c$.

The recursion terminates. A child retaining its parameter loses a fixed fraction of the inherited mass, while the other children halve the parameter. Mass is bounded below by one for a nonempty set, and a parameter can be halved only $\OO(\log(d_0/W))=\OO(L)$ times before the base case. Every node has finitely many children.

Induction now proves correctness of Algorithm 3. Successful child calls supply exactly the potentials required by \cref{thm:reconstruction}. Successful reconstruction returns a valid parent potential. A failure gives a negative $w$-cycle by \cref{lem:cycle-certificate}. At the root, a valid $c$-potential supplies the scaling step.

\subsection{Running time}

\begin{proof}[Proof of \cref{thm:one-scale}]
Consider the actual recursion tree of Algorithm 3. Its internal nodes have the potential decreases $\Pi_X$ from \eqref{eq:actual-drop}. For a complete successful execution, these decreases telescope exactly, giving
\begin{equation}\label{eq:telescope}
 \sum_{X\text{ internal}}\Pi_X
 =\Psi(M_0,d_0)-\sum_{Y\text{ terminal}}\Psi(M(Y),d_Y)
 \le\Psi(M_0,d_0)=\OO(M_0L).
\end{equation}
The root bound follows from $\log(4d_0/W)=1+2\log M_0$. Terminal parameters are at least $W/4$, so their potentials are nonnegative. The same inequality holds for an execution stopped by a negative cycle: telescope over the expanded nodes and retain all pending children as boundary leaves. Their potentials are also nonnegative.

By \cref{thm:decomp}, constructing the decomposition of $X$ costs $\OO(M(X)L)$ and $\Pi_X=\Omega(M(X))$. By \cref{thm:reconstruction}, merging the child potentials costs $\OO(L\Pi_X)$. Thus all internal work is $\OO(L\sum_X\Pi_X)=\OO(M_0L^2)$. Total child incidence is
\[
 \OO\!\left(\sum_{X\text{ internal}}M(X)\right)=\OO(M_0L).
\]
It pays for creating induced child graphs, maintaining membership data, and scanning terminal instances. The return search of \cref{lem:cycle-certificate} adds only $\OO(M_0)$, since the algorithm ends at its first certificate.
\end{proof}

\begin{proof}[Proof of \cref{thm:main}]
Apply \cref{thm:one-scale} in the scaling reduction of \cref{lem:scaling}. There are $\OO(\log(nW))$ scaling steps and one final nonnegative SSSP computation. Substituting $M_0=\Theta(m+n\log\log n)$ and $L=\Theta(\log n)$ gives the stated bound. All procedures use deterministic choices.
\end{proof}

\end{document}